\documentclass[journal]{IEEEtran}

\IEEEoverridecommandlockouts

\usepackage{color}
\usepackage{subfigure}
\usepackage{graphics} 
\usepackage{epsfig} 
\usepackage{mathptmx} 

\usepackage{amsmath} 
\usepackage{amssymb}  
\usepackage{psfrag}
\usepackage{bm}
\usepackage{amsbsy}
\newtheorem{thm}{Theorem}

\newtheorem{lemma}[thm]{Lemma}

\title{ Optimal Delay-Throughput Trade-offs in Mobile Ad-Hoc Networks:\\
Hybrid Random Walk and One-Dimensional Mobility Models\thanks{An earlier version of this paper
appeared in the Proc. of ITA Workshop, 2007.}\thanks{Research presented here was supported in part
by a Vodafone Fellowship and NSF grant CNS 05-19691.}}

\author{Lei Ying and R. Srikant\\
Coordinated Science Lab\\ and\\ Department of Electrical and Computer Engineering\\
University of Illinois at Urbana-Champaign\\
\{lying, rsrikant\}@uiuc.edu}

\begin{document}
\maketitle
\begin{abstract}
Optimal delay-throughput trade-offs for two-dimensional i.i.d mobility models have been established
in \cite{YinYanSri_06}, where we showed that the optimal trade-offs can be achieved using rate-less
codes when the required delay guarantees are sufficient large. In this paper, we extend the results
to other mobility models including two-dimensional hybrid random walk model, one-dimensional i.i.d.
mobility model and one-dimensional hybrid random walk model. We consider both fast mobiles and slow
mobiles, and establish the optimal delay-throughput trade-offs under some conditions. Joint
coding-scheduling algorithms are also proposed to achieve the optimal trade-offs.
\end{abstract}

\section{Notations}
The following notations are used throughout this paper, given non-negative functions $f(n)$ and
$g(n)$:
\begin{enumerate} \item[(1)]
$f(n)=O(g(n))$ means there exist positive constants $c$ and $m$ such that $f(n) \leq cg(n)$ for all
$ n\geq m$.

\item[(2)] $f(n)=\Omega(g(n))$ means there exist positive constants $c$ and $m$ such that $f(n)\geq
cg(n)$ for all $n\geq m.$ Namely, $g(n)=O(f(n))$.

\item[(3)] $f(n)=\Theta(g(n))$ means  that both $f(n)=\Omega(g(n))$ and $f(n)=O(g(n))$ hold.

\item[(4)]$f(n)=o(g(n))$ means that $\lim_{n\rightarrow \infty} f(n)/g(n)=0.$

\item[(5)] $f(n)=\omega(g(n))$ means that $\lim_{n\rightarrow \infty} g(n)/f(n)=0.$ Namely,
$g(n)=o(f(n)).$

\end{enumerate}

\section{Introduction}
Delay-throughput trade-offs in mobile ad-hoc networks have received much attention since the work
of Grossglauser and Tse \cite{GroTse_01}, where they showed that the throughput of ad-hoc networks
can be significantly improved by exploring the node mobility. Recently the trade-off was
investigated under different mobility models, which include the i.i.d. mobility \cite{NeeMod_05,
TouGol_04, LinShr_04, YinYanSri_06}, one-dimensional mobility \cite{DigGroTse_02,
GamMamParSha_06_2}, random walk \cite{GamMamPraSha_04, GamMamParSha_06, GamMamParSha_06_1,
ShaMazShr_06}, hybrid random walk \cite{ShaMazShr_06} and Brownian motion \cite{LinShaMazShr_06}.

In \cite{YinYanSri_06}, we demonstrated that the optimal trade-offs for two-dimensional i.i.d.
mobility models can be achieved using rate-less codes when the required delay guarantees are
sufficiently large. In this paper, we extend the results to the two-dimensional hybrid random walk,
one-dimensional i.i.d. mobility and one-dimensional hybrid random walk models. The two-dimensional
i.i.d. mobility studied in \cite{YinYanSri_06} only models the case where the network topology
changes dramatically at each time slot. However Markovian mobility dynamics may be more realistic.
Thus the two-dimensional hybrid random walk model was introduced by Sharma \emph{et al} in
\cite{ShaMazShr_06}, where the unit square is divided into $1/S^2$ small-squares, and mobiles move
from the current small-square to one of its eight adjacent small-squares at the beginning of each
time slot (The detailed description of the two-dimensional hybrid random walk model is presented in
Section \ref{sec: model}). Since the distance each mobile can move is at most $2\sqrt{2}/S$ at each
time slot, we can use different values of $S$ to model mobiles with different speeds, so this
two-dimensional hybrid random walk model can be used for a wide range of scenarios. Note that the
two-dimensional hybrid random walk model is the same as the two-dimensional i.i.d. mobility model
when $S=1.$ One might wonder why the results in \cite{YinYanSri_06} are necessary given the results
in this paper. The reason is that the Markovian mobility dynamics in this paper requires a
different set of tools than those in \cite{YinYanSri_06} and as a result, the trade-off in this
paper are applicable only when $S=o(1).$ Thus, the results in \cite{YinYanSri_06} cannot be
recovered from the results of this paper. We wish to comment that one of the main differences
between this paper and \cite{YinYanSri_06} is that, the i.i.d. mobility assumption in
\cite{YinYanSri_06} allows us to use Chernoff bounds to obtain concentration results. However, the
random walk and other mobility models in this paper require the use of martingale inequalities to
establish the travel patterns of the mobiles.

In this paper, we will also study one-dimensional mobility models. These models are motivated by
certain types of delay-tolerant networks \cite{War_03}, in which a satellite sub-network is used to
connect local wireless networks outside of the Internet. Since the satellites move in fixed orbits,
they can be modelled as one-dimensional mobilities on a two-dimensional plane. Motivated by such a
delay-tolerant network, we consider one-dimensional mobility model where $n$ nodes move
horizontally and the other $n$ node move vertically. Since the node mobility is restricted to one
dimension, sources have more information about the positions of destinations compared with the
two-dimensional mobility models. We will see that the throughput is improved in this case; for
example, under the one-dimensional i.i.d. mobility model with fast mobiles, the trade-off will be
shown to be $\Theta(\sqrt[3]{D^2/n}),$ which is better than $\Theta(\sqrt{D/n}),$ the trade-off
under the two-dimensional i.i.d. mobility model with fast mobiles. We also propose joint
coding-scheduling algorithms which achieve the optimal trade-offs.

Three mobility models are included in this paper, and each model will be investigated under both
the fast-mobility and slow-mobility assumptions. The detailed analysis of the two-dimensional
hybrid random walk model and one-dimensional i.i.d. mobility model will be presented. The results
of the one-dimensional hybrid random walk model can be obtained using the techniques used in the
other two models, so the analysis is omitted in this paper for brevity. Our main results include
the followings:
\begin{enumerate}
\item[(1)] Two-dimensional hybrid random walk model:
\begin{enumerate}
\item[(i)] Under the fast mobility assumption, it is shown that the maximum throughput per S-D pair
is $O(\sqrt{D/n})$ when $S=o(1)$ and $D=\omega(|\log S|/S^2),$ and Joint Coding-Scheduling
Algorithm I \cite{YinYanSri_06} can achieve the maximum throughput when $S=o(1)$ and $D$ is both
$\omega(\max\{(\log^2 n)|\log S|/S^6, \sqrt[3]{n}\log n\})$ and $o(n/\log^2 n).$

\item[(ii)] Under the slow mobility assumption, it is shown that the maximum throughput per S-D
pair is $O(\sqrt[3]{D/n})$ when $S=o(1)$ and $D=\omega(|\log S|/S^2),$ and Joint Coding-Scheduling
Algorithm II can achieve the maximum throughput when $S=o(1)$ and $D$ is both $\omega((\log^2
n)|\log S|/S^6)$ and $o(n/\log^2 n).$
\end{enumerate}

\item[(2)] One-dimensional i.i.d. mobility model:

\begin{enumerate}
\item[(i)] Under the fast mobility assumption, it is shown that the maximum throughput per S-D is
$O\left(\sqrt[3]{D^2/n}\right)$ given delay constraint $D.$ Then Joint Coding-Scheduling Algorithm
III is proposed to achieve the maximum throughput when $D$ is both $\omega(\sqrt[5]{n})$ and
$o\left(\sqrt{n}/\sqrt[3/2]{\log n}\right).$

\item[(ii)] Under the slow mobility assumption, it is shown that the maximum throughput per S-D
pair is $O\left(\sqrt[4]{D^2/n}\right).$ Joint Coding-Scheduling Algorithm IV is proposed to
achieve the maximum throughput when $D$ is $o\left(\sqrt{n}/\log^2 n\right).$
\end{enumerate}

\item[(3)] One-dimensional hybrid random walk model:

\begin{enumerate}
\item[(i)] Under the fast mobility assumption, it is shown that the maximum throughput per S-D pair
is $O(\sqrt[3]{D^2/n})$ when $S=o(1)$ and $D=\omega(1/S^2),$ and Joint Coding-Scheduling Algorithm
III can achieve the maximum throughput when $S=o(1)$ and $D$ is both $\omega(\max\{(\log^2 n)|\log
S|/S^4, \sqrt[5]{n}\log n\})$ and $o(\sqrt{n}/\sqrt[3/2]{\log n}).$

\item[(ii)] Under the slow mobility assumption, it is shown that the maximum throughput per S-D
pair is $O(\sqrt[4]{D^2/n})$ when $S=o(1)$ and $D=\omega(1/S^2),$ and Joint Coding-Scheduling
Algorithm IV can achieve the maximum throughput when $S=o(1)$ and $D$ is both $\omega((\log^2
n)|\log S|/S^4)$ and $o(\sqrt{n}/\log^2 n).$
\end{enumerate}

\end{enumerate}
Note that the optimal delay-throughput trade-off are established under some conditions on $D.$ When
these conditions are not met, the trade-off is still unknown in general, though a trade-off of the
two-dimensional hybrid random walk model with slow mobiles has been established under an assumption
regarding packet replication in \cite{ShaMazShr_06}. We also would like to mention that when the
step size of the two-dimensional hybrid random walk is $1/\sqrt{n},$ our two-dimensional hybrid
random walk model is identical to the random walk model studied in \cite{GamMamParSha_06,
GamMamParSha_06_1}, where the optimal delay-throughput trade-off has been obtained. Our results do
not apply to this case since the set of allowed values for $D$ becomes empty in that case (see (1)
(i) above).

The remainder of the paper is organized as follows: In Section \ref{sec: model}, we introduce the
communication and mobility model. Then we analyze the two-dimensional hybrid random walk models in
Section \ref{sec: RW}, and one-dimensional i.i.d. mobility models in Section \ref{sec: OD_STSM}.
The results of one-dimensional hybrid random walk model are presented in Section \ref{sec: ORW}.
Finally, the conclusions is given in Section \ref{sec: conl}.

\section{Model}
\label{sec: model} In this section, we first present the models that we use for mobility and
wireless interference. Then the definitions of delay and throughput are provided.

\noindent{\bf Mobile Ad-Hoc Network Model:} Consider an ad-hoc network where wireless mobile nodes
are positioned in a unit square. Assume that the time is slotted, we study following three mobility
models in this paper.
\begin{enumerate}
\item[(1)]{\bf Two-Dimensional Random Walk Model:} Consider a unit square which is further divided
into $1/S^2$ squares of equal size. Each of the smaller square will be called an RW-cell (random
walk cell), and indexed by $(U^x, U^y)$ where $U^x, U^y\in \{1, \ldots, 1/S\}.$ The unit square is
assumed to be a torus, i.e., the top and bottom edges are assumed to touch each other and similarly
the left and right edges also are assumed to touch other. A node which is in one RW-cell at a time
slot moves to one of its eight adjacent RW-cells or stays in the same RW-cell in the next time-slot
with each move being equally likely as in Figure \ref{fig: 2DRW}. Two RW-cells are said to be
adjacent if they share a common point. The node position within the RW-cell is randomly uniformly
selected. There are $n$ S-D pairs in the network. Each node is both a source and a destination.
Without loss of generality, we assume that the destination of node $i$ is node $i+1,$ and the
destination of node $n$ is node $1.$
\begin{figure}[hbt]
\centering{\epsfig{file=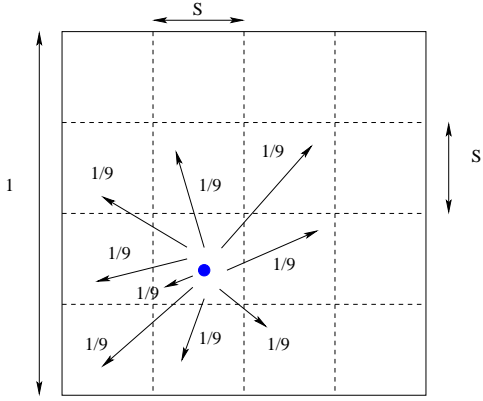,width=2.5in}} \caption{Two-Dimensional Random Walk Model}
\label{fig: 2DRW}
\end{figure}

\item[(2)]{\bf One-Dimensional I.I.D. Mobility Model:} Our one-dimensional i.i.d. mobility model is
defined as follows:
\begin{enumerate}
\item[(i)] There are $2n$ nodes in the network. Among them, $n$ nodes, named H-nodes,  move
horizontally; and the other $n$ nodes, named V-nodes, move vertically.

\item[(ii)] Using $(x_i, y_i)$ to denote the position of node $i.$ If node $i$ is an H-node, $y_i$
is fixed and $x_i$ is a value randomly uniformly chosen from $[0, 1].$ We also assume that H-nodes
are evenly distributed vertically, so $y_i$ takes values $1/n, 2/n, \ldots, 1.$ V-nodes have
similar properties.

\item[(iii)] Assume that source and destination are the same type of nodes. Also assume that node
$i$ is an H-node if $i$ is odd, and a V-node if $i$ is even. Further, assume that the destination
of node $i$ is node $i+2,$ the destination of node $2n-1$ is node $1,$ and the destination of node
$2n$ is node $2.$

\item[(iv)] The orbit distance of two H(V)-nodes is defined to be the vertical (horizontal)
distance of the two nodes.
\end{enumerate}

\item[(3)]{\bf One-Dimensional Random Walk Model:} Each orbit is divided into $1/S$ RW-intervals
(random walk interval). At each time slot, a node moves into one of two adjacent RW-intervals or
stays at the current RW-interval (see Figure \ref{fig: 1DRW}). The node position in the RW-interval
is randomly, uniformly selected.
\begin{figure}[hbt]
\centering{\epsfig{file=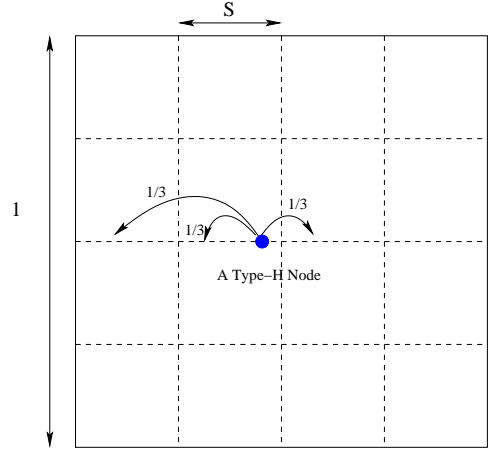,width=2.5in}} \caption{One-Dimensional Random Walk Model}
\label{fig: 1DRW}
\end{figure}

\end{enumerate}

\noindent{\bf Communication Model:} We assume the protocol model introduced in \cite{gupkum00} in
this paper. Let $\hbox{dist}(i,j)$ denote the Euclidean distance between node $i$ and node $j,$ and
$r_i$ to denote the transmission radius of node $i.$ A transmission from node $i$ can be
successfully received at node $j$ if and only if following two conditions hold:
\begin{enumerate}
\item[(i)] $\hbox{dist}(i,j)\leq r_i;$

\item[(ii)] $\hbox{dist}(k,j)\geq (1+\Delta) \hbox{dist}(i,j)$ for each node $k\not=i$ which
transmits at the same time, where $\Delta$ is a protocol-specified guard-zone to prevent
interference.
\end{enumerate}
We further assume that at each time slot, at most $W$ bits can be transmitted in a successful
transmission.

\noindent{\bf Time-Scale of Mobility:}  Two time-scales of mobility are considered in this paper.
\begin{enumerate}
\item[(1)] Fast mobility: The mobility of nodes is at the same time-scale as the data transmission,
so $W$ is a constant independent of $n$ and only one-hop transmissions are feasible in single time
slot.

\item[(2)] Slow mobility: The mobility of nodes is much slower than the wireless transmission, so
$W\gg n.$ Under this assumption, the packet size can be scaled as $W/H(n)$ for $H(n)=O(n)$ to
guarantee $H(n)$-hop transmissions are feasible in single time slot.

\end{enumerate}

\noindent{\bf Delay and Throughput:} We consider hard delay constraints in this paper. Given a
delay constraint $D,$ a packet is said to be successfully delivered if the destination obtains the
packet within $D$ time slots after it is sent out from the source.

Let $\Lambda_i[T]$ denote the number of bits successfully delivered to the destination of node $i$
in time interval $[0, T].$ A throughput of $\lambda$ per S-D pair is said to be feasible under the
delay constraint $D$ and loss probability constraint $\epsilon>0$ if there exists $n_0$ such that
for any $n\geq n_0,$ there exists a coding/routing/scheduling algorithm with the property that each
bit transmitted by a source is received at its destination with probability at least $1-\epsilon,$
and
\begin{equation} \lim_{T\rightarrow \infty}\Pr\left( \displaystyle
\frac{\Lambda_i[T]}{T}\geq \lambda, \hbox{ }\forall \hbox{ } i\right)=1. \label{eq: throughput}
\end{equation}

\section{Two-Dimensional Hybrid Random Walk Models}
\label{sec: RW} The optimal delay-throughput trade-offs of the two-dimensional i.i.d. mobility
model with fast mobiles and slow mobiles have been established in \cite{YinYanSri_06}. In this
section, we first first extend the results to two-dimensional hybrid random walk models. We will
obtain the maximum throughput for $D=\omega\left(|\log S|/S^2\right),$ and then show that the
maximum throughput can be achieved using the algorithms proposed in \cite{YinYanSri_06} under some
additional constraints on $D.$
\subsection{Upper Bound}
The upper bound is established under the following assumptions:

\noindent{\bf Assumption 1:} Packets destined for different nodes cannot be encoded together.

\noindent{\bf Assumption 2:} A new coded packet is generated right before the packet is sent out.
The node generating the coded packet does not store the packet in its buffer.

\noindent{\bf Assumption 3:} Once a node receives a packet (coded or uncoded), the packet is not
discarded by the node till its deadline expires.

Note that Assumption $1$ is the only significant restriction imposed on coding/routing/scheduling
schemes. Next we introduce following notations which will be used in our proof.
\begin{itemize}
\item $\Lambda[T]:$ $\Lambda[T]=\sum_{i=1}^n\Lambda_i[T].$

\item $b:$ Index of a bit stored in the network. Bit $b$ could be either a bit of a data packet or
a bit of a coded packet.

\item $d_b:$ The destination of bit $b.$

\item $c_b:$ The node storing bit $b.$

\item $t_b:$ The time slot at which bit $b$ is generated.

\item $\tilde{L}_{b}:$ The minimum distance between node $d_b$ and node $c_b$ from time slot $t_b$
to time slot $t_b+D-1,$ i.e.,
$$\tilde{L}_b=\min_{t_b\leq t\leq t_b+D-1} \hbox{dist}(d_b, c_b)(t).$$
\end{itemize}

\begin{thm} Consider the two-dimensional hybrid random walk model with step-size $S=o(1)$ and delay constraint
$D=\omega(|\log S|/S^2),$ and suppose that Assumption 1-3 hold. We have following results:
\begin{enumerate}
\item[(1)] For fast mobiles,
\begin{eqnarray}\frac{48\sqrt{2}WT}{\Delta\sqrt{\pi}}\sqrt{n}(\sqrt{D}+1)\geq E[\Lambda[T]].\label{eq:
fast}\end{eqnarray}

\item[(2)] For slow mobiles,
\begin{eqnarray}\frac{8\sqrt[3]{9}WT}{\sqrt[3/2]{\Delta\pi^2}}\sqrt[3/2]{n}(\sqrt[3]{D}+1)\geq E[\Lambda[T]].\label{eq: slow}\end{eqnarray}
\end{enumerate}
\label{thm: Upper_hrw}
\end{thm}
\begin{proof} Let $N^{\rm rw}_{b}$ denote the number of time slots, from
$t_{b}+1$ to $t_{b}+D,$ at which node $c_b$ and $d_b$ are in the same RW-cell or neighboring
RW-cells. Then for any $L\in[0, S/\sqrt{\pi}),$ we have
\begin{eqnarray*}
&&\Pr\left(\tilde{L}_b\leq L \right)\\
&=&\sum_{K=1}^{D} \Pr\left(\tilde{L}_b\leq L
|{N}^{\rm rw}_b=K \right)\Pr\left(N^{\rm rw}_b=K\right)\\
&\leq& \sum_{K=1}^{D} \left(1-\left(1-\frac{\pi L^2}{S^2}\right)^K\right)\Pr\left(N^{\rm
rw}_b=K\right)\\
&=& 1-E\left[\left(1-\frac{\pi L^2}{S^2}\right)^{N^{\rm rw}_b}\right]\\
&\leq & 1-\left(1-\frac{\pi L^2}{S^2}\right)^{E\left[N^{\rm rw}_b\right]},
\end{eqnarray*} where the first inequality follows from the fact that the node position within a
RW-cell is randomly uniformly selected, and the last inequality follows from the Jensen's
inequality.

Next we consider $E[N_b^{\rm tw}].$ Let $(U_i^x(t), U_i^y(t))$ denote the RW-cell in which node $
i$ is at time slot $t,$ and $(V^x_i(t), V^y_i(t))$ denote the displacement of node $i$ at time slot
$t,$ i.e.,
\begin{eqnarray*}
V^x_i(t)=\left\{%
\begin{array}{ll}
    1, & \hbox{w.p. } \frac{1}{3} \\
    0, & \hbox{w.p. } \frac{1}{3} \\
    -1, & \hbox{w.p. } \frac{1}{3} \\
\end{array}%
\right.\hbox{ and } V^y_i(t)=\left\{%
\begin{array}{ll}
    1, & \hbox{w.p. } \frac{1}{3} \\
    0, & \hbox{w.p. } \frac{1}{3} \\
    -1, & \hbox{w.p. } \frac{1}{3} \\
\end{array}%
\right..
\end{eqnarray*} It is easy to see that
\begin{eqnarray*}
U_i^x(t)&=&\left[\left(U_i^x(0)+\sum_{m=1}^{t-1} V^x_i(m)\right)\mod{\frac{1}{S}}\right]+1;\\
U_i^y(t)&=&\left[\left(U_i^y(0)+\sum_{m=1}^{t-1} V^y_i(m)\right)\mod{\frac{1}{S}}\right]+1.
\end{eqnarray*}
Further, let $(U_{i-j}^x(t), U_{i-j}^y(t))$ denote the relative position of node $i$ from node $j,$
i.e.,
\begin{eqnarray*}
U_{i-j}^x(t)&=&\left(U_{i-j}^x(0)+\sum_{m=1}^{t-1} \tilde{V}^x_{i-j}(m)\right)\mod{\frac{1}{S}};\\
U_{i-j}^y(t)&=&\left(U_{i-j}^y(0)+\sum_{m=1}^{t-1} \tilde{V}^y_{i-j}(m)\right)\mod{\frac{1}{S}},
\end{eqnarray*} where
\begin{eqnarray*}
\tilde{V}^x_{i-j}(t)={V}^x_{i}(t)-{V}^x_{j}(t)=\left\{%
\begin{array}{ll}
    2, & \hbox{w.p. } \frac{1}{9} \\
    1, & \hbox{w.p. } \frac{2}{9}\\
    0, & \hbox{w.p. } \frac{1}{3}\\
    -1, & \hbox{w.p. } \frac{2}{9}\\
    -2, & \hbox{w.p. } \frac{1}{9}\\
\end{array}%
\right.,
\end{eqnarray*} and
\begin{eqnarray*}
\tilde{V}^y_{i-j}(t)={V}^y_{i}(t)-{V}^y_{j}(t)=\left\{%
\begin{array}{ll}
    2, & \hbox{w.p. } \frac{1}{9} \\
    1, & \hbox{w.p. } \frac{2}{9}\\
    0, & \hbox{w.p. } \frac{1}{3}\\
    -1, & \hbox{w.p. } \frac{2}{9}\\
    -2, & \hbox{w.p. } \frac{1}{9}\\
\end{array}%
\right..
\end{eqnarray*} So $(U_{i-j}^x(t), U_{i-j}^y(t))$ is the consequence of random walk $(\tilde{V}^x_{i-j}(m),
\tilde{V}^y_{i-j}(m))$ with initial position $$(U_{i-j}^x(0), U_{i-j}^y(0))=(U_i^x(0)-U_j^x(0),
U_i^y(0)-U_j^y(0)).$$

Note that node $c_b$ and node $d_b$ are in the same RW-cell if $(U_{c_b-d_b}^x(t),
U_{c_b-d_b}^y(t))=(0,0),$ and in neighboring RW-cells if $(U_{c_b-d_b}^x(t),
U_{c_b-d_b}^y(t))\in\{(0,1), (1,0), (1,1), (0, 1/S-1), (1/S-1, 0), (1/S-1, 1/S-1) \}.$ Similar to
the argument in Lemma \ref{lem: RWM} provided in Appendix B, we can conclude that for
$D=\omega(|\log S|/S^2),$
\begin{eqnarray*} E\left[N^{\rm rw}_b\right]\leq \frac{99}{10}
S^2 D,
\end{eqnarray*} which implies that
\begin{eqnarray}
\Pr\left(\tilde{L}_b\leq L \right) &\leq & 1-\left(1-\frac{\pi
L^2}{S^2}\right)^{\frac{99}{10}S^2 D}\nonumber\\
&\leq & 36 L^2 D. \label{eq: hp}
\end{eqnarray}

Based on inequality (\ref{eq: hp}), the proof of inequality (\ref{eq: fast}) is similar to the
proof of Theorem 3 of \cite{YinYanSri_06}, and the proof of (\ref{eq: slow}) is similar to the
proof of Theorem 6 of \cite{YinYanSri_06}.
\end{proof}

\subsection{Joint Coding-Scheduling Algorithms}
From Theorem \ref{thm: Upper_hrw}, we can see that the optimal delay-throughput trade-offs of the
two-dimensional hybrid random walk models are similar to the ones of the two-dimensional i.i.d.
mobility models \cite{YinYanSri_06}. It motivates us to consider the algorithms proposed in
\cite{YinYanSri_06}. As in \cite{YinYanSri_06}, we define and categorize packets into four
different types.
\begin{itemize}
\item Data packets: There are the uncoded data packets that have to be transmitted by the sources
and received by the destinations.

\item Coded packets: Packets generated by Raptor codes. We let $(i, k)$ denote the $k^{\rm th}$
coded packet of node $i.$

\item Duplicate packets: Each coded packet could be broadcast to other nodes to generate multiple
copies, called duplicate packets. We let $(i,k,j)$ denote a copy of $(i,k)$ carried by node $j,$
and $(i,k, J)$ to denote the set of all copies of coded packet $(i,k).$

\item Deliverable packets: Duplicate packets that happen to be within distance $L$ from their
destinations.
\end{itemize}
We will show that the optimal trade-offs can be achieved using Joint Coding-Scheduling Algorithm I
and II presented in \cite{YinYanSri_06} with the following modifications:
\begin{enumerate}\item[(1)] For the fast mobility model, we use Joint Coding-Scheduling Algorithm I with the following modification:
$2D/(25M)$ data packets are coded into $D/M$ coded packets;

\item[(2)] For the slow mobility model, we use Joint Coding-Scheduling Algorithm II with the
following modifications: $D/7$ data packets are coded to $D$ coded packets.
\end{enumerate}
For the detail of the algorithms, please refer to \cite{YinYanSri_06}.

\begin{thm} Consider the two-dimensional hybrid random walk models.
\begin{enumerate}
\item[(1)] Fast mobility model:  Suppose that $S$ is $o(1),$ $D$ is both $\omega(\max\{\log^2
n|\log S|/(S^6), \sqrt[3]{n}\log n\})$ and $o(n/(\log^2 n)),$ and the delay constraint is $6D.$
Then under the fast mobility model, given any $\epsilon$ there exists $n_0$ such that for any
$n\geq n_0,$ every data packet sent out can be recovered at the destination with probability at
least $1-\epsilon,$ and
\begin{eqnarray}\lim_{T\rightarrow \infty} \Pr\left(\frac{\Lambda_i[T]}{T}\geq
\left(\frac{W}{2520}\right)\left(\sqrt{\frac{D}{n}}\right) \hbox{ } \forall i\right)=1.\label{eq:
th}
\end{eqnarray} by using the modified Joint Coding-Scheduling
Algorithm I.

\item[(2)] Slow mobility model: Suppose that $S$ is $o(1)$ and $D$ is both $\omega(\log^2 n|\log
S|/(S^6))$ and $o(n/(\log^3 n)),$ and the delay constraints is $16D.$ Then under the slow mobility
model, given any $\epsilon$ there exists $\tilde{n}_0$ such that for any $n\geq \tilde{n}_0,$ every
data packet sent out can be recovered at the destination with probability at least $1-\epsilon,$
and
\begin{eqnarray}\lim_{T\rightarrow \infty} \Pr\left(\frac{\Lambda_i[T]}{T}\geq
\left(\frac{W}{224\sqrt{2}c_s C}\right)\left(\sqrt[3]{\frac{D}{n}}\right) \hbox{ } \forall
i\right)=1.\label{eq: th2}
\end{eqnarray}
by using the modified Joint Coding-Scheduling Algorithm II.
\end{enumerate}
\label{thm: lower_hrw}
\end{thm}

\begin{proof} Let $\tilde{A}$ denote the area of a cell,
and $\tilde{M}[t]$ to denote the number of nodes in the cell at time slot $t.$ A cell is said to be
a \emph{good cell} at time $t$ if
$$\frac{9}{10} \tilde{A} n+1\leq \tilde{M}[t]\leq \frac{11}{10}{\tilde{A}n}.$$

\noindent{\bf Proof of (1):}  We consider one super time slot which consists of $6D$ time slots,
and calculate the probability that the $2D/(25M)$ data packets from node $i$ are fully recovered at
the destination, where $M=\sqrt{n/D}$ is the mean number of nodes in each cell. The proof will show
the following events happen with high probability.

\noindent{\bf Node distribution:} All cells are good during the entire super-time-slot with high
probability. Letting $\cal G$ denote this event, we will show
\begin{eqnarray}\Pr\left({\cal G}\right)\geq 1-\frac{1}{n^2}. \label{eq: good}
\end{eqnarray}

\noindent{\bf Broadcasting:} At least $16D/(25M)$ coded packets from a source are successfully
duplicated after the broadcasting step with high probability, where a coded packet is said to be
successfully duplicated if the packet is in at least $4M/5$ distinct relay nodes. Letting $A_i$
denote the number of coded packets which are successfully duplicated in a super time slot, we will
first show that
\begin{eqnarray}\Pr\left(\left. A_i\geq \frac{16D}{25M}\right| {\cal G}\right)\geq
1-\frac{55D}{n}-e^{-\frac{D}{600M}}. \label{eq: A}
\end{eqnarray}

\noindent{\bf Receiving:} At least $3D/(25M)$ distinct coded packets from a source are delivered to
its destination after the receiving step with high probability. Letting $B_i$ denote the number of
distinct coded packets delivered to destination $i+1$ in a super time slot, we will show
\begin{eqnarray}
\Pr\left(\left.B_{i}\geq \frac{3}{25}\frac{D}{M}\right|  A_i\geq \frac{16}{25}
\frac{D}{M}\right)\geq 1-2e^{-\frac{\log D}{8500}}- e^{-\frac{D}{500M\log D}}. \label{eq: B}
\end{eqnarray}

From inequalities (\ref{eq: good}), (\ref{eq: A}) and (\ref{eq: B}), we can conclude that under the
Joint Coding-Scheduling Algorithm I, at each super time slot, the $2D/(25M)$ data packets can be
successfully recovered with probability at least
$$1-\frac{1}{n^2}-\frac{55D}{n}-e^{-\frac{D}{600M}}-2e^{-\frac{\log D}{8500}}- e^{-\frac{D}{500M\log D}}.$$ The
rest of the proof is the same as the proof of Theorem 4 in \cite{YinYanSri_06}.

{\bf Analysis of node distribution:} Since $D=o\left(n/\log^2 n\right)$ implies $M=\omega(\log n).$
Inequality (\ref{eq: good}) can be obtained from the Chernoff bound and union bound.

{\bf Analysis of broadcasting step: } Consider the broadcasting step. Note that when ${\cal G}$
occurs, node $i$ is selected to broadcast with probability at least $10/(11M)$ at each time slot.
Let ${\cal B}_i[t]$ denote the event that node $i$ is selected to broadcast in time slot $t.$ From
the Chernoff bound, we have
\begin{eqnarray}
\Pr\left(\left.\sum_{t=1}^D 1_{{\cal B}_{i}[t]}\geq \frac{9}{11}\frac{D}{M}\right| {\cal G}
\right)\geq 1-e^{-\frac{D}{600M}}. \label{eq: br}
\end{eqnarray}
So node $i$ broadcasts $9D/(11M)$ coded packets with a high probability. Each coded packet is
broadcast to $9M/10$ relay nodes.

According to Step (2)(ii) of Joint Coding-Scheduling Algorithm I \cite{YinYanSri_06}, each relay
node keeps at most one packet for each source. Consider duplicate packet $(i, k, j).$ It could be
dropped if node $j$ is in the same cell as node $i$ and node $i$ is selected to broadcast. Thus,
the probability that $(i, k, j)$ is dropped is at most
\begin{eqnarray}\frac{11}{10}\frac{DM}{n} \times \frac{10}{9}\frac{1}{M} =\frac{11}{9}\frac{D}{n}
\label{eq: p}
\end{eqnarray} due to the following two facts:
\begin{enumerate}
\item[(a)] Let ${\cal H}_{ji}[D]$ denote the event that node $j$ is in the same cell as node $i$ in
at least one of $D$ consecutive time slots. Similar to (\ref{eq: hp}), it can be shown that
\begin{eqnarray}
 \Pr\left({{\cal H}_{ji}[D]}\right)\leq \frac{11}{10}\frac{D M}{n}\label{eq: p1}
\end{eqnarray}
under the delay constraint given in the theorem.

\item[(b)] When ${\cal G}$ occurs, node $i$ is selected to broadcast with probability at most
$10/(9M)$ at each time slot.
\end{enumerate}

Now suppose source $i$ broadcasts $\tilde{D}_i$ coded packets, so $9M\tilde{D}_i/10$ duplicate
copies are generated.  Let $\tilde{N}^d_i$ denote the number of duplicate packets of node $i$
dropped in the broadcasting step. From the Markov inequality and inequality (\ref{eq: p}), we have
\begin{eqnarray*}
&&\Pr\left(\tilde{N}_i^d\geq \frac{M\tilde{D}_i}{50}\left|{\cal G}, \sum_{t=1}^D 1_{{\cal
B}_i[t]}=\tilde{D}_i\right.\right)\\
&\leq& \frac{E\left[\tilde{N}_i^d\geq \frac{M\tilde{D}_i}{50}\left|{\cal G}, \sum_{t=1}^D 1_{{\cal
B}_i[t]}=\tilde{D}_i\right.\right]}{\frac{M\tilde{D}_i}{50}}\\
&\leq&\frac{\frac{9M\tilde{D}_i}{10}\times
\frac{11D}{9n}}{\frac{M\tilde{D}_i}{50}}\\
&=&\frac{55D}{n},
\end{eqnarray*} which implies
\begin{eqnarray}\Pr\left(A_i\geq \frac{4}{5}\tilde{D}_i\left|{\cal G}, \sum_{t=1}^D 1_{{\cal B}_i[t]}=\tilde{D}_i\right.\right)\geq 1-\frac{55D}{n}\label{eq:
NC}\end{eqnarray} since otherwise, more than $M\tilde{D}_i/50$ duplicate packets would be dropped.
Inequality (\ref{eq: A}) follows from inequalities (\ref{eq: NC}) and (\ref{eq: br}).

{\bf Analysis of receiving step:} We group every $3D/\log D$ time slots into big time slots, named
as b-time-slot and indexed by $t_b,$ and then divided every b-time-slot into three equal parts,
indexed by $t_{b,1},$ $t_{b, 2}$ and $t_{b, 3}$ as in Figure \ref{fig: TD}.
\begin{figure}[hbt]\psfrag{t1}{$\scriptstyle t_{b,1}$}\psfrag{t2}{$\scriptstyle t_{b,2}$}\psfrag{t3}{$\scriptstyle t_{b,3}$}
\centering{\epsfig{file=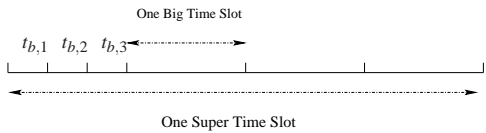,width=2.5in}} \caption{The Division of A Super Time Slot}
\label{fig: TD}
\end{figure}
We first calculate the probability that coded packet $(i, k)$ is delivered in $t_{b, 2}.$  Let
${\cal H}_{(i,k)}[t_{b, 2}]$ denote the event that at least one copy of packet $(i,k)$ becomes
deliverable in $t_{b,2}.$ If $(i,k)$ is in at least $4M/5$ relay nodes, we can obtain
\begin{eqnarray} \Pr\left({{\cal H}_{(i,k)}\left[t_{b, 2}\right]}\right)&\geq& 1-\left(\left(1-\frac{
M}{nS^2}\right)^{\frac{4DS^2}{5\log D}} +\frac{1}{n}\right)^{\frac{4M}{5}}\nonumber\\
&\geq& \frac{3}{5}\frac{1}{\log D} \label{eq: p2}
\end{eqnarray} due to the following facts:
\begin{enumerate}
\item[(a)] Given $D/\log D=\omega(\log n|\log S|/S^6),$ from Lemma \ref{lem: RWM} provided in
Appendix B, we know that with probability at least $1-1/n,$ two nodes are in the same RW-cell for
at least $4DS^2/(5\log D)$ time slots.

\item[(b)] Given two nodes are in the same RW-cell, the probability that they are in the same cell
is $M/(nS^2).$
\end{enumerate}
Next note that the duration of $t_{b, 1}$ and $t_{b,3}$ are of a larger order than the mixing time
of the random walk (the mixing time is defined in Appendix B). From the definition of the mixing
time, we have that at any time slot belonging to $t_{b, 2},$ the nodes are almost uniformly
distributed in the unit square. Let ${\cal D}_{(i,k)}[t_{b, 2}]$ denote the event that coded packet
$(i,k)$ is delivered to its destination in $t_{b, 2}.$ Following the argument used to prove
inequality (13) of Theorem 4 in \cite{YinYanSri_06}, we have
\begin{eqnarray}\Pr\left({{\cal
D}_{(i,k)}[t_{b, 2}]}\right)\geq \frac{3}{20\log D}.\label{eq: PD}\end{eqnarray}

Now let $\mathbf{x}_{t}$ denote the positions of the nodes at time slot $t,$ and
$$\tilde{\mathbf{X}}=\{{\mathbf{x}}_t\}_{t=\frac{3(k-1)D}{\log D}+1}$$ for $k=1, \ldots, 5\log D/3.$
Also let ${\cal D}_{(i,k)}$ denote the event that $(i,k)$ is delivered in the receiving step. It is
easy to see that ${\cal D}_{(i,k)}$ occurs if ${\cal D}_{(i,k)}[t_{b,2}]$ occurs for some $t_{b,
2}\in\{1, \ldots, 5\log D/3\}.$  Note that $\{{\cal D}_{(i, k)}[t_{b, 2}]\}_{t_{b, 2}}$ are
mutually independent given $\tilde{\mathbf{X}},$ so from inequality (\ref{eq: PD}), we have
$$\Pr\left(\left.{{\cal D}_{(i,k)}}\right|\tilde{\mathbf{X}}\right)\geq 1-\left(1-\frac{3}{20}\frac{1}{\log D}\right)^{5\log D/3}\geq \frac{1}{5}.$$
Further since $B_i\geq \sum_{k=1}^{A_i} 1_{{\cal D}_{(i,k)}},$ we can conclude that
\begin{eqnarray}
E\left[B_i\left|\tilde{\mathbf{X}}, A_i\geq \frac{16}{25}\frac{D}{M}\right.\right]\geq
\frac{16}{125} \frac{D}{M}. \label{eq: mean}
\end{eqnarray}

We next bound the number of distinct coded packets deliverable in $t_b.$ Similar to inequality
(\ref{eq: p2}), we have $$\Pr\left({{\cal H}_{(i,k)}\left[t_b\right]}\right)\leq \frac{3}{\log
D}.$$ Note that no two duplicate packets from node $i$ are in one relay node, so $\{{\cal
H}_{(i,k)}[t_b]\}_k$ are mutually independent. From the Chernoff bound, we have
$$\Pr\left(\sum_{k=1}^{D/M} 1_{{\cal
H}_{(i,k)}\left[t_b\right]}  \leq \frac{16}{5} \frac{D}{M\log D}\right)\geq 1-e^{-\frac{D}{400M\log
D}}.$$ Let $\tilde{{\cal F}}_i$ denote the event that node $i$ obtains no more than $16 D/(5M\log
D)$ coded packets at each b-time-slot in the receiving step. From the union bound, we have that for
sufficiently large $n,$
\begin{eqnarray}
\Pr\left({\tilde{{\cal F}}_i} \right) \geq 1- \left(\frac{5}{3}\log D\right)
\left(e^{-\frac{D}{400M\log D}}\right)\geq 1- e^{-\frac{D}{500M\log D}}. \label{eq: bound}
\end{eqnarray}
Now let $B_i(\tilde{\mathbf{X}}, A_i, {\cal F}_i)$ denote the number of distinct coded packets
delivered to the destination of node $i$ given $(\tilde{\mathbf{X}}, A_i, {\cal F}_i),$ and
$\mathbf{X}_{t_b}$ denote an $n\times (3D/\log D)$ matrix where the $(i, t)$ entry is the position
of node $i$ at the $t^{\rm th}$ time slot of b-time-slot $t_b.$ It is easy to see that the value of
$B_i(\tilde{\mathbf{X}}, A_i, {\cal F}_i)$ is determined by $\{\mathbf{X}_{t_b}\},$ i.e., there
exists a function $f_{(\tilde{\mathbf{X}}, A_i, {\cal F}_i)}$ such that
$$B_i(\tilde{\mathbf{X}}, A_i, {\cal F}_i)=f_{(\tilde{\mathbf{X}}, A_i, {\cal F}_i)}\left(\mathbf{X}_1,\ldots, \mathbf{X}_{5\log D/3}\right).$$
From the definition of ${\cal F}_i,$ function $f_{(\tilde{\mathbf{X}}, A_i, {\cal F}_i)}$ satisfies
the following condition,
\begin{eqnarray}\left|f_{(\tilde{\mathbf{X}}, A_i, {\cal F}_i)}\left(\mathbf{X}_1,\ldots, \mathbf{X}_{t_b-1}, \mathbf{X}_{t_b},
\mathbf{X}_{t_b+1},\ldots, \mathbf{X}_{5\log
D/3}\right)-\right.\nonumber&\\\left.f_{(\tilde{\mathbf{X}}, A_i, {\cal
F}_i)}\left(\mathbf{X}_1,\ldots, \mathbf{X}_{t_b-1}, \mathbf{Y}_{t_b}, \mathbf{X}_{t_b+1},\ldots,
\mathbf{X}_{5\log D/3}\right)\right|\leq& \frac{16}{5} \frac{D}{M\log D}.\label{eq:
AHB}\end{eqnarray}

It is easy to see that $\{\mathbf{X}_{t_b}\}$ are mutually independent given $(\tilde{\mathbf{X}},
A_i,  \tilde{{\cal F}}_i).$ Then invoking Azuma-Hoeffding inequality provided in Appendix A, we can
conclude that
\begin{eqnarray}
&\Pr\left(B_{i}{(\tilde{\mathbf{X}}, A_i, {\cal F}_i)}\geq E\left[B_i\left|\tilde{\mathbf{X}}, A_i,
\tilde{{\cal F}}_i \right.\right]-\frac{1}{125}\frac{D}{M}\right)\nonumber\\ \geq&
1-2e^{-\frac{\log D}{8500}} \label{eq: AHE}
\end{eqnarray} holds for any $\tilde{\mathbf{X}}$ and $A_i.$ Inequality (\ref{eq: B}) follows from inequalities  (\ref{eq: mean}), (\ref{eq:
bound}) and (\ref{eq: AHE}).

\noindent{\bf Proof of (2):} We consider one super time slot which consists of $16D$ time slots,
and calculate the probability that the $D/7$ data packets from node $i$ are fully recovered at the
destination. Let $M_1=\sqrt[3]{n/D},$ which is the mean number of nodes in each cell at the
broadcasting step. Following the analysis above, we can prove that the following events happen with
high probability.

\noindent{\bf Node distribution:} All cells are good during the entire super-time-slot with high
probability, i.e.,
\begin{eqnarray*}
\Pr\left({\cal G}\right)\geq 1-\frac{1}{n^2}.
\end{eqnarray*}

\noindent{\bf Broadcasting:} At least $4D/5$ coded packets from a source are successfully
duplicated after the broadcasting step with high probability. Specifically, we have
\begin{eqnarray}
\Pr\left(\left.A_i\geq \frac{4}{5}D\right| {\cal G}\right)\geq 1-\frac{50}{M_1^2},
\label{eq: A_s}
\end{eqnarray} where a coded packet is said to be successfully duplicated if it is in
$4M_1/5$ distinct relay nodes.

\noindent{\bf Receiving:} At least $D/6$ distinct coded packets from a source are delivered to its
destination after the receiving step with high probability. Specifically, we have
\begin{eqnarray}
\Pr\left(\left.B_i\geq \frac{D}{6} \right| A_i\geq \frac{4}{5} D\right)\geq 1-2e^{-\frac{\log
D}{800}}. \label{eq: B_s}
\end{eqnarray}

From inequalities (\ref{eq: good}), (\ref{eq: A_s}) and (\ref{eq: B_s}), we can conclude that under
the modified Joint Coding-Scheduling Algorithm II, at each super time slot, the $D/7$ source
packets can be successfully recovered with probability at least
$$1-\frac{1}{n^2}-\frac{50}{M_1^2}-2e^{-\frac{\log
D}{800}},$$ and theorem holds.
\end{proof}

\section{One-Dimensional I.I.D. Mobility Models}
\label{sec: OD_STSM}
\subsection{Upper Bounds}

\begin{thm}
Consider the one-dimensional i.i.d. mobility model, and assume that Assumption 1-3 hold. We have
following results:
\begin{enumerate}
\item[(1)] For fast mobiles,
$$8WT\sqrt[3]{\frac{2}{\pi \Delta^2}}\sqrt[3/2]{n}(\sqrt[3/2]{D}+1)\geq E[\Lambda[T]];$$

\item[(2)] For slow mobiles,
$$4WT\sqrt[4]{\frac{4}{\pi \Delta^2}}\sqrt[3/4]{n}(\sqrt{D}+1)\geq E[\Lambda[T]].$$
\end{enumerate}
\label{thm: OD_upper}
\end{thm}
\begin{proof}
Recall that $\tilde{L}_b$ is the minimum distance between node $d_b$ and node $c_b$ from time slot
$t_b$ to $t_b+D-1.$ If the orbits of node $c_b$ and $d_b$ are vertical to each other, then
$\tilde{L}_b\leq L$ holds only if at some time slot $t,$ node $c_b$ and $d_b$ are in the square
with side length $2L$ as in Figure \ref{fig: Hit_V}. In this case, we have
$$\Pr\left(\tilde{L}_b\leq L\right)\leq 1-(1-4L^2)^D.$$ If the orbits of node $c_b$ and $d_b$ are
parallel to each other, then it is easy to verify that $$\Pr\left(\tilde{L}_b\leq L\right)\leq
1-(1-2L)^D.$$

Thus for $L\leq 1/2,$ we can conclude that$$\Pr\left(\tilde{L}_b\leq L\right)\leq 1-(1-2L)^D\leq
2LD.$$ The rest of the proof is similar to Theorem 3 and Theorem 6 in \cite{YinYanSri_06}.
\begin{figure}[hbt]
\psfrag{2L}{$\scriptstyle 2L$} \psfrag{node i}{\scriptsize node $i$}\psfrag{node j}{\scriptsize
node $j$}\centering{\epsfig{file=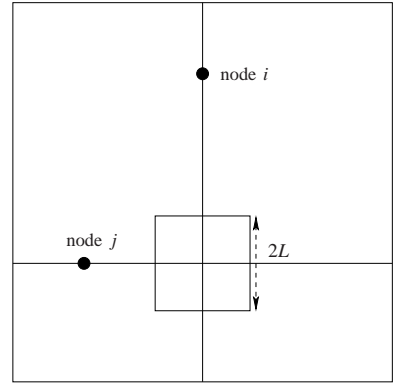,width=2in}} \caption{The Two Orbits are Vertical to
Each Other} \label{fig: Hit_V}
\end{figure}
\end{proof}

\subsection{Joint Coding-Scheduling Algorithm for Fast Mobility}
Choose $$M=\sqrt[3]{\frac{n}{D^2}}.$$ We divide the unit square into $\sqrt{n/M}$ horizontal
rectangles, named as H-rectangles; and $\sqrt{n/M}$ vertical rectangles, named as V-rectangles as
in Figure \ref{fig: Rec}. A packet is said to be destined to a rectangle if the orbit of its
destination is contained in the rectangle.
\begin{figure}[hbt]
\psfrag{S}{$\scriptstyle \sqrt{M/n}$} \centering{\epsfig{file=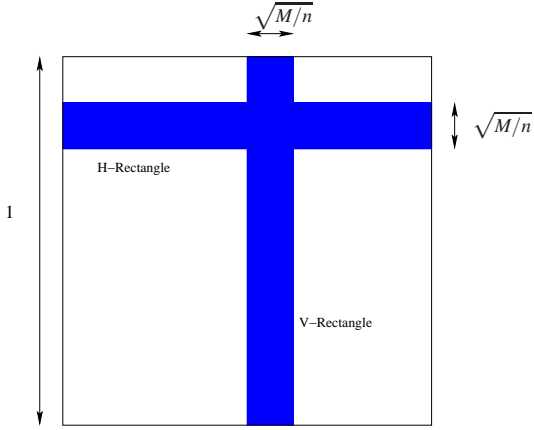,width=2.5in}} \caption{The
Unit Square is Divided into H(V)-Rectangles} \label{fig: Rec}
\end{figure}

The algorithms for the one-dimensional i.i.d. mobility model has four steps. The first step is the
Raptor encoding. The second step is the broadcasting step. In this step, the H(V)-nodes broadcast
coded packets to V(H)-nodes. The third step is the transporting step, where the V(H)-nodes
transport the H(V)-packets to the H(V)-rectangles containing the orbits of corresponding
destinations, and then broadcast packets to the H(V)-nodes whose orbits are contained in the
rectangles. After the third step, all duplicate packets are carried by the nodes that move parallel
with the destinations and their orbit distance is less than $\sqrt{M/n}.$ The fourth step is the
receiving step, where the packets are delivered to the destinations.

Since duplicate copies are generated in both the broadcasting step and the transporting step. To
distinguish them, we name the duplicate packets generated at the broadcasting step as B-duplicate
packets, and the duplicate packets generated at the transporting step as T-duplicate packets. Also
we say a B-duplicate packet is transportable if it is in the rectangle containing the orbit of the
destination of the packet.

Consider a cell with area $\tilde{A}$ and use $\tilde{M}^{H(V)}[t]$ to denote the number of
H(V)-nodes in the cell. For the one-dimensional mobility model, a cell is said to be a \emph{good
cell} at time slot $t$ if
$$\frac{9}{10}\tilde{A}n+1\leq \tilde{M}^{H(V)}[t]\leq \frac{11}{10}\tilde{A}n.$$

Next we present the Joint Coding-Scheduling Algorithm III, which achieves the maximum throughput
obtain in Theorem \ref{thm: OD_STSM_L}. Note that in the following algorithm, each time slot is
further divide into $C$ mini-time slots, and each cell is guaranteed to be active in at least one
of mini-time slot within each time slot.

\noindent{\bf Joint Coding-Scheduling Algorithm III:} The unit square is divided into a regular
lattice with $n/M$ cells, and the packet size is chosen to be $W/(2C).$ We group every $7D$ time
slots into a super time slot. At each super time slot, the nodes transmit packets as follows.
\begin{enumerate}
\item[(1)] {\bf Raptor Encoding:} Each source takes $2D/(35M)$ data packets, and uses the Raptor
codes to generate $D/M$ coded packets.

\item[(2)]  {\bf Broadcasting:} This step consists of $D$ time slots. At each time slot, the nodes
execute the following tasks:
\begin{enumerate}
\item[(i)] In each good cell, one H-node and one V-node are randomly selected. If the selected
H(V)-node has never been in the current cell before and not already transmitted all of its $D/M$
coded packets, then it broadcasts a coded packet that was not previous transmitted to $9M/10$
V(H)-nodes in the cell during the mini-time slot allocated to that cell.

\item[(ii)] All nodes check the duplicate packets they have. If more than one B-duplicate packets
 are destined to the same rectangle, randomly keep one and drop the others.
\end{enumerate}

\item[(3)]{\bf Transporting:} This step consists of $D$ time slots. At each time slot, the nodes do
the following:
\begin{enumerate}
\item[(i)] Suppose that node $j$ is a V-node ,and carries B-duplicate packet $(i,k,j).$ Node $j$
broadcasts $(i,k,j)$ to $9M/10$ H-nodes in the same cell if following conditions hold: (a) Node $j$
is in a good cell; (b) B-duplicate packet $(i,k,j)$ is the only transportable H-packet in the cell.

\item[(ii)] Each node checks the T-duplicate packets it carries. If more than one T-duplicate
packet has the same destination, randomly keep one and drop the others.
\end{enumerate}

\item[(4)]{\bf Receiving:} This step consists of $5D$ time slots.  At each time slot, if there are
no more than two deliverable packets in the cell, the deliverable packets are delivered to the
destinations with one-hop transmissions. At the end of this step, all undelivered  packets are
dropped. The destinations decode the received coded packets using Raptor decoding.
\end{enumerate}

\begin{thm}
Consider Joint Coding-Scheduling Algorithm III. Suppose $D$ is $\omega(\sqrt[5]{n})$ and
$o\left(\sqrt{n}/\sqrt[3/2]{\log n}\right),$ and the delay constraint is $7D.$ Then given any
$\epsilon>0,$ there exists $n_0$ such that for any $n\geq n_0,$ every data packet sent out can be
recovered at the destination with probability at least $1-\epsilon,$ and furthermore
\begin{eqnarray*}\lim_{T\rightarrow \infty} \Pr\left(\frac{\Lambda_i[T]}{T}\geq
\left(\frac{W}{980C}\right)\sqrt[3]{\frac{D^2}{n}} \hbox{ } \forall i\right)=1.
\end{eqnarray*}
\label{thm: OD_STSM_L}
\end{thm}
\begin{proof}
Consider one super time slot and let $\cal G$ denote the event that all cells are good in the super
time slot. The proof will show the following events happen with high probability.

\noindent{\bf Node distribution:} All cells are good during the entire super-time-slot with high
probability. Specifically, it is easy to verify that
\begin{eqnarray}\Pr\left({\cal G}\right)\geq 1-\frac{1}{n^2}. \label{eq: f}
\end{eqnarray}

\noindent{\bf Broadcasting:} At least $2D/(3M)$ coded packets from a source are successfully
duplicated after the broadcasting step with high probability, where a coded packet is said to be
successfully duplicated if it has at least $4M/5$ B-duplicate packets. Specifically, we will show
\begin{eqnarray}
\Pr\left(\left.A_i\geq \frac{2}{3}\frac{D}{M}\right|{\cal G}\right) \geq 1- \frac{40}{M}.
\label{eq: 31}
\end{eqnarray}

\noindent{\bf Transporting:} At least $9D/(70M)$ coded packets from a source are successfully
transported after the transporting step with high probability, where a coded packet is said to be
successfully transported if it has at least $4M/5$ T-duplicate copies. Letting $C_i$ denote the
number of successfully transported packets from node $i,$ we will show
\begin{eqnarray}
\Pr\left(\left. C_i\geq \frac{9}{70}\frac{D}{M} \right| {\cal G},  A_i\geq \frac{2}{3}\frac{D}{M}
\right) \geq 1- 3e^{-\frac{D}{1500M}}-\frac{100}{M^2}. \label{eq: 33}
\end{eqnarray}

\noindent{\bf Receiving:} At least $9D/(140M)$ distinct coded packets from a source are delivered
to its destination after the receiving step. Specifically, we will show
\begin{eqnarray}
\Pr\left(\left. B_i\geq \frac{9}{140} \frac{D}{M}\right| C_i\geq \frac{9}{70} \frac{D}{M} \right)
\geq 1- 2e^{-\frac{D}{8000 M}}. \label{eq: 34}
\end{eqnarray}

If the claims above hold, the probability that the $D/10M$ data packets are fully recovered in one
super time slot is at least
$$1-\frac{1}{n^2}-\frac{40}{M}- 3e^{-\frac{D}{1500M}}-\frac{100}{M^2}- 2e^{-\frac{D}{8000 M}}.$$
The theorem follows from a similar argument provided in Theorem 4 of \cite{YinYanSri_06}.

{\bf Analysis of broadcasting step:} Assume that ${\cal G}$ occurs, then at each time slot, node
$i$ is selected with probability $10/(11M).$ Note that there are $\sqrt{n/M}$ cells on the orbit of
node $i,$ and node $i$ is uniformly randomly positioned in one of the cells. Thus, the number of
coded packets broadcast by node $i$ is equal to the number of non-empty bins of following
balls-and-bins problem.

\noindent\emph{Balls-and-Bins Problem:} Suppose we have $\sqrt{n/M}$ bins and one trash can. At
each time slot, we drop a ball. Each bin receives the ball with probability $10/(11\sqrt{nM}),$ and
the trash can receives the ball with probability $1-10/(11M).$ Repeat this $D$ times, i.e., $D$
balls are dropped.

\noindent From Lemma \ref{lem: ball-bin} provided in Appendix A, we have
$$\Pr\left(\left.\sum_{k=1}^{D} 1_{{\cal B}_{i}[t]} \geq \frac{9D}{11M}\right|{\cal G}\right)\geq 1-e^{-\frac{3D}{1000M}}.$$

We say two nodes are competitive with each other if  the orbits of their destinations are contained
in the same rectangle, so each node has $\sqrt{Mn}-1$ competitive nodes. Suppose that node $i$ is
an H-node and node $j$ is a V-node. Let $\tilde{N}^c_{i,j}(t)$ denote the number of node $i$'s
competitive nodes in the V-rectangle containing the orbit of node $j$ at time slot $t.$ Since nodes
are uniformly, randomly positioned on their orbits, from the Chernoff bound, we have
\begin{eqnarray}\Pr\left(\tilde{N}_{i,j}^c(t)\leq \frac{11}{10}M\right)\geq 1-e^{-\frac{M}{300}}.
\label{eq: DP}
\end{eqnarray}
Now consider B-duplicate packet $(i,k,j)$ and assume that node $z,$ a competitive of node $i,$ is
in the V-rectangle containing the orbit of node $i.$ Then $(i,k,j)$ might be dropped when it is in
the same cell as node $z,$ and node $z$ is selected to broadcast. The probability of this event is
at most
\begin{eqnarray}
\sqrt{\frac{M}{n}}\times\frac{10}{9M}. \label{eq: PH}
\end{eqnarray}
From (\ref{eq: DP}), (\ref{eq: PH}) and the union bound, we can conclude that the probability that
$(i, k, j)$ is dropped at time slot $t$ is at most
\begin{eqnarray}
e^{-\frac{M}{300}}+\frac{11}{10}M\times\sqrt{\frac{M}{n}}\times\frac{10}{9M}
=e^{-\frac{M}{300}}+\frac{11}{9}\sqrt{\frac{M}{n}},\label{eq: NoC}
\end{eqnarray}
which implies that the probability of $(i, k, j)$ dropped in the broadcasting step is at most
\begin{eqnarray*}
1-\left(1-e^{-\frac{M}{300}}-\frac{11}{9}\sqrt{\frac{M}{n}}\right)^D \leq
De^{-\frac{M}{300}}+\frac{11}{9M}
\end{eqnarray*}  Inequality (\ref{eq: 31}) follows from above inequality and the Markov inequality.

{\bf Analysis of transporting step:} Consider an H-node $i.$ Let $C_{(i,k)}$ denote the number of
B-duplicate packets which are contained in the V-rectangle where $(i,k)$ broadcast, and are
destined to the same H-rectangle as node $i.$ Note the following facts:
\begin{enumerate}
\item[(a)] Each node has $\sqrt{Mn}-1$ competitive nodes.

\item[(b)] Each H-node broadcasts at most $D/M$ coded packets. The probability that a coded packet
broadcast in a specific V-rectangle is at most
$$\frac{D}{M}\times\sqrt{\frac{M}{n}}=\frac{D}{\sqrt{nM}}.$$

\item[(c)] Each broadcast generates $9M/10$ duplicate copies.
\end{enumerate}
Thus, from the Chernoff bound, we have that
$$\Pr\left(C_{(i,k)}\leq DM\right)\geq 1-e^{-\frac{D}{300M}},$$ which implies that for sufficiently large $n,$
\begin{eqnarray}\Pr\left(C_{(i,k)}\leq DM \hbox{ } \forall \hbox{ } k\right)\geq
1-\frac{D}{M}e^{-\frac{D}{300M}}\geq 1-e^{-\frac{D}{400M}}.\label{eq: 34_1}\end{eqnarray} Let
${\cal T}_{(i,k)}$ denote the event that a B-duplicate packet is broadcast at time slot $t$ in the
transporting step. If $(i,k)$ is successfully duplicated, i.e., there are at least $4M/5$
B-duplicate copies of $(i,k),$ we have
$$\Pr\left({\cal T}_{(i,k)}(t)\right)\geq
\frac{4M}{5}\sqrt{\frac{M}{n}}\left(1-\sqrt{\frac{M}{n}}\right)^{C_{(i,k)}+\frac{9M}{10}-1}.$$
Further, let ${\cal T}_{(i,k)}$ denote the event that at least one copy of $(i,k)$ is broadcast in
the transporting step. Then for sufficiently large $n,$ we can obtain that
\begin{eqnarray*}&&\Pr\left({{\cal T}_{(i,k)}}\left|C_{(i,k)}\leq DM\right.\right) \\
&\geq& 1-\left(1-
\frac{4M}{5}\sqrt{\frac{M}{n}}\left(1- \sqrt{\frac{M}{n}}\right)^{DM+9M/10-1} \right)^{D} \\
&\geq& \frac{1}{4}. \end{eqnarray*} Let $C^b_i$ denote the number of distinct coded packets of node
$i$ broadcast in the transporting step, i.e., $$C^b_i=\sum_{k=1}^{D/M} 1_{{\cal T}_{(i,k)}}.$$
Since different coded packets of node $i$ are broadcast in different V-rectangles, $\{{\cal
T}_{(i,k)}\}$ are mutually independent. From the Chernoff bound, we have
\begin{eqnarray}\Pr\left(\left.C^b_i\geq \frac{D}{7 M}\right| A_i\geq \frac{2D}{3M}, C_{(i,k)}\leq MD \hbox{ }
\forall k\right)\geq 1-2e^{-\frac{D}{900M}}.\label{eq: 34_2}\end{eqnarray}

In the transporting step, a T-duplicate copy will be dropped if the node carrying it obtains
another packet destined to the same destination. Consider a T-duplicate packet $(i,k,l)$ carried by
node $l.$ Note following facts:
\begin{enumerate}
\item[(a)] Coded packets of node $i$ are broadcast in at most $D/M$ V-rectangles.

\item[(b)] Each rectangle contains at most $9M/10$ B-duplicate copies from node $i.$
\end{enumerate}
Thus, the probability of $(i,k,l)$ dropped at time slot $t$ is at most
$$\frac{D}{M}\sqrt{\frac{M}{n}}\left(1-\sqrt{\frac{M}{n}}\right)^{\frac{9M}{10}}$$
The node mobility is independent across time, so the probability of $(i,k,l)$ dropped in the
transporting step is at most
$$1-\left(1-\frac{D}{M}\sqrt{\frac{M}{n}}\left(1-\sqrt{\frac{M}{n}}\right)^{\frac{9M}{10}}\right)^D \leq \frac{1}{M^2}.$$

Let $\bar{N}_i^d$ denote the number of duplicate packets dropped in the transporting step. Note
that $9MC_i^b/10$ T-duplicate packets are generated, and each of them has probability $1/M^2$ to be
dropped. Using the Markov inequality, we have
$$\Pr\left(\bar{N}_i^d\geq \frac{M C_i^b}{100}\right)\leq \frac{90}{M^2},$$ which implies
\begin{eqnarray}\Pr\left(C_i\geq \frac{9}{10} C_i^b \right)\geq 1-\frac{90}{M^2}\label{eq: 34_3}\end{eqnarray} since
otherwise, more than $MC_i^b/100$ duplicate copies are dropped. Inequality (\ref{eq: 33}) follows
form inequality (\ref{eq: 34_1})-(\ref{eq: 34_3}).

{\bf Analysis of receiving step:} The proof is similar to the proof of inequality (13) in
\cite{YinYanSri_06}.
\end{proof}

\subsection{Joint Coding-Scheduling Algorithm for Slow Mobility}
In this subsection, we propose an algorithm which achieves the delay-throughput trade-off obtained
in Theorem \ref{thm: OD_upper}. First choose
\begin{eqnarray*}
M_1&=&\sqrt[4]{\frac{n}{4D^2}}\\
M_2&=&M_1^2
\end{eqnarray*}
and scale the packet size to be
$$\frac{W}{4c_s M_1}$$ where $c_s$ is a constant independent of $n$ as in \cite{YinYanSri_06}.
Further, we divide the unit square into $\sqrt{n/M_2}$ horizontal rectangles, named as
H-rectangles; and $\sqrt{n/M_2}$ vertical rectangles, named V-rectangles.

\noindent{\bf Joint Coding-Scheduling Algorithm IV:} We group every $14D$ time slots into a super
time slot. At each super time slot, the packets are coded and transmitted as follows:
\begin{enumerate}
\item[(1)] {\bf Raptor Encoding:} Each source takes $D/50$ data packets, and uses the Raptor codes
to generate $D$ coded packets.

\item[(2)]  {\bf Broadcasting:} The unit square is divided into a regular lattice with $n/M_1$
cells. This step consists of $D$ time slots. At each time slot, the nodes execute the following
tasks:
\begin{enumerate}
\item[(i)] The nodes in good cells take their turns to broadcast. If node $i$ is a H(V)-node and
has never been in the current cell before, it randomly selects $9M_1/10$ V(H)-nodes and broadcasts
a coded packet to them.

\item[(ii)] Each node checks the B-duplicate packets it carries. If there are multiple B-duplicate
packets destined to the same rectangle, randomly pick one and drop the others.
\end{enumerate}

\item[(3)]{\bf Transporting:} The unit square is divided into a regular lattice with $n/M_1$ cells.
This step consists of $2D$ time slots. At each time slot, the nodes do the following:
\begin{enumerate}
\item[(i)] Suppose node $j$ carries duplicate packet $(i,k,j),$ which is an H-packet. If node $j$
is in a good cell and and $(i,k,j)$ is deliverable, node $j$ broadcasts the packet to $9M_1/10$
H-nodes in the cell.

\item[(ii)] Each node checks the T-duplicate packets it carries. if there is more than one
T-duplicate packet destined to the same destination, randomly pick one and drop the others.
\end{enumerate}

\item[(4)]{\bf Receiving:} The unit square is divided into a regular lattice with $n/M_2$ cells.
This step consists of $12D$ time slots. At each time slot, the nodes in good cells do the following
at the mini-time slot allocated to their cells:

\begin{enumerate}
\item[(i)] The nodes which contain deliverable packets randomly pick one deliverable packet and
send a request to the corresponding destination.

\item[(ii)] For each destination, it accepts only one request.

\item[(iii)] The nodes whose requests are accepted transmit the deliverable packets to their
destinations using the highway algorithm proposed in \cite{FraDouTseThi_05}.
\end{enumerate}
At the end of this step, all undelivered duplicate packets are dropped.  Destinations use Raptor
decoding to decode the received coded packets.
\end{enumerate}

\begin{thm}
Consider Joint Coding-Scheduling Algorithm IV. Suppose $D$ is both $\omega(1)$ and
$o\left(\sqrt{n}/\log^2 n\right),$ and the delay constraint is $14D.$ Then given any $\epsilon>0,$
there exists $n_0$ such that for any $n\geq n_0,$ every data packet sent out can be recovered at
the destination with probability $1-\epsilon,$ and furthermore\begin{eqnarray*}\lim_{T\rightarrow
\infty} \Pr\left(\frac{\Lambda_i[T]}{T}\geq \left(\frac{W}{1400\sqrt{2} c_s
C}\right)\sqrt[4]{\frac{D^2}{n}} \hbox{ } \forall i\right)=1.
\end{eqnarray*}
\label{thm: OD_TSSM_L}
\end{thm}
\begin{proof} Following the analysis of Theorem \ref{thm: OD_STSM_L}, we can show the following events happen with
high probability.

\noindent{\bf Node distribution:} All cells are good during the entire super-time-slot with high
probability, i.e.,
\begin{eqnarray}
\Pr\left({\cal G}\right)\geq 1-\frac{1}{n^2}. \label{eq: 111}
\end{eqnarray}

\noindent{\bf Broadcasting:} At least $3D/10$ coded packets from a source are successfully
duplicated after the broadcasting step with high probability, where a coded packet is said to be
successfully duplicated if it has at least $M_1/3$ B-duplicate packets. Specifically, we have
\begin{eqnarray}
\Pr\left(\left.A_i\geq \frac{3}{10} D\right| {\cal G}\right) \geq 1- \frac{1}{n^2}. \label{eq: 11}
\end{eqnarray}

\noindent{\bf Transporting:} At least $3D/40$ coded packets from a source are successfully
transported after the transporting step with high probability, where a coded packet is said to be
successfully transported if it has at least $4M_1/5$ T-duplicate copies. Specifically, we have
\begin{eqnarray}
\Pr\left(\left. C_i\geq \frac{3}{40}D\right|{\cal G}, A_i\geq \frac{3}{10}D \right) \geq 1-
e^{-\frac{D}{3600}}-\frac{180}{\log n}. \label{eq: 12}
\end{eqnarray}

\noindent{\bf Receiving:} At least $D/40$ distinct coded packets from a source are delivered to its
destination after the receiving step. Specifically, we have
\begin{eqnarray}
\Pr\left(\left. B_i\geq \frac{D}{40}\right|{\cal G}, C_i\geq \frac{3}{40}D\right) \geq 1-
2e^{-\frac{D}{1000}}. \label{eq: 13}
\end{eqnarray}

Thus, the probability that the $D/50$ data packets are fully recovered in one super time slot is at
least
$$1-\frac{2}{n^2}-
e^{-\frac{D}{3600}}-\frac{180}{\log n}- 2e^{-\frac{D}{1000}},$$ and theorem holds.
\end{proof}

\section{One-Dimensional Hybrid Random Walk Model, Fast Mobiles and Slow Mobiles}
\label{sec: ORW} In this section, we present the optimal delay-throughput trade-offs of the
one-dimensional hybrid random walk model. The results can be proved following the analysis of the
one-dimensional i.i.d. mobility and the analysis of the two-dimensional hybrid random walk. The
details are omitted here for brevity.

\begin{thm}
Consider the one-dimensional hybrid random walk model and assume that Assumption 1-3 hold. Then for
$S=o(1)$ and $D=\omega(1/S^2),$ We have following results:
\begin{enumerate}
\item[(1)] For fast mobiles, \begin{eqnarray}24WT\sqrt[3]{\frac{1}{\pi
\Delta^2}}\sqrt[3/2]{n}(\sqrt[3/2]{D}+1)\geq E[\Lambda[T]].\label{eq: T1D}\end{eqnarray}

When $S=o(1)$ and $D$ is both $\omega(\max\{(\log^2 n)|\log S|/S^4, \sqrt[5]{n}\log n\})$ and
$o(\sqrt{n}/\sqrt[3/2]{\log n}),$ Joint Coding-Scheduling Algorithm III can be used to achieve a
throughput same as (\ref{eq: T1D}) except for a constant factor.

\item[(2)]For slow mobiles, \begin{eqnarray}12WT\sqrt[4]{\frac{1}{\pi
\Delta^2}}\sqrt[3/4]{n}(\sqrt{D}+1)\geq E[\Lambda[T]].\label{eq: T2D}\end{eqnarray}

When  $S=o(1)$ and $D$ is both $\omega((\log^2 n)|\log S|/S^4)$ and $o(\sqrt{n}/\log^2 n),$ Joint
Coding-Scheduling Algorithm IV can be used to achieve a throughput same as (\ref{eq: T2D}) except
for a constant factor.
\end{enumerate}

\rightline{$\square$}
\end{thm}

\section{Conclusion}
\label{sec: conl} In this paper, we investigated the optimal delay-throughput trade-off of a mobile
ad-hoc network under the two-dimensional hybrid random walk,  one-dimensional i.i.d. mobility model
and one-dimensional hybrid random walk model. The optimal trade-offs have been established under
some conditions on delay $D.$ When these conditions are not met, the optimal trade-offs are still
unknown in general. We would like to comment that the key to establishing the optimal
delay-throughput trade-off is to obtain $P_{i,j}(D, L),$ the probability that node $i$ hits node
$j$ in one of $D$ consecutive time slots given a hitting distance $L.$ For example, under the
two-dimensional hybrid random walk model, the upper bound was obtained under the condition
$D=\omega(|\log S| /S^2)$ since it was the condition under which we established an upper bound on
$P_{i,j}(D, L)$ (inequality (\ref{eq: hp})). Further, the maximum throughput was shown to be
achievable under a more restrict condition $D=\omega((\log^2 n)|\log S| /S^6)$ since it was the
condition under which we established a lower bound on $P_{i,j}(D, L)$ (inequality (\ref{eq: p2})).
Thus, if we can find techniques to compute $P_{i,j}(D, L)$ without using the restricts on $D,$ then
the delay-throughput trade-offs can be characterized more generally. This is a topic for future
research.

\textbf{Acknowledgment:} We thank Sichao Yang for useful discussions during the course of this
work.

\section*{Appendix A: Probability Results}
The following lemmas are some basic probability results.
\begin{lemma}
Let $X_1,\ldots, X_n$ be independent $0-1$ random variables such that $\sum_i X_i=\mu.$ Then, the
following Chernoff bounds hold
\begin{eqnarray}
\hbox{Pr}\left(\sum_{i=1}^n X_i<(1-\delta)\mu\right)&\leq & e^{-\delta^2\mu/2};\label{eq: ChB_1}\\
\hbox{Pr}\left(\sum_{i=1}^n X_i>(1+\delta)\mu\right)& \leq &e^{-\delta^2\mu/3}.\label{eq: ChB_2}
\end{eqnarray}
\label{lem: C_B}
\end{lemma}
\begin{proof}
A detailed proof can be found in \cite{MitUpf_05}.
\end{proof}

\begin{lemma}
Assume we have $m$ bins. At each time, choose $h$ bins and drop one ball in each of them. Repeat
this $n$ times. Using $N_1$ to denote the number of bins containing at least one ball, the
following inequality holds for sufficiently large $n.$
\begin{eqnarray}
\hbox{Pr}\left(N_1\leq (1-\delta) m\tilde{p}_1 \right)&\leq& 2e^{-\delta^2
m\tilde{p}_1/3}.\label{eq: BB_2}
\end{eqnarray}
where $\tilde{p}_1=1-e^{-\frac{nh}{m}}.$ \label{lem: ball-bin_2}
\end{lemma}
\begin{proof}
Please refer to \cite{YinYanSri_06} for a detailed proof.
\end{proof}

\begin{lemma}
Suppose $n$ balls are independently dropped into $m$ bins and one trash can. After a ball is
dropped, the probability in the trash can is $1-p,$ and the probability in a specific bin is $p/m.$
Using $N_2$ to denote the number of bins containing at least $1$ ball, the following inequality
holds for sufficiently large $n.$
\begin{eqnarray}
\hbox{Pr}\left(N_2\leq (1-\delta) m \tilde{p}_2\right)&\leq& 2e^{-\delta^2 m
\tilde{p}_2/3};\label{eq: BB_1}
\end{eqnarray}
where $\tilde{p}_2=1-e^{-\frac{np}{m}}.$ \label{lem: ball-bin}
\end{lemma}
\begin{proof}
Please refer to \cite{YinYanSri_06} for a detailed proof.
\end{proof}

Next we introduce the Azuma-Hoeffding inequality.
\begin{lemma}
Suppose that $X_0,\ldots, X_n$ are independent random variables, and there exists a constant $c>0$
such that $f(\mathbf{X})=f(X_1, \ldots, X_n)$
satisfies the following condition for any $i$ and any set of values $x_1, \ldots, x_n$ and $y_i:$   \begin{eqnarray*} \left|f(x_1, \ldots, x_{i-1}, x_i, x_{i+1}, \ldots, x_n)-\right.&\\
\left.f(x_1, \ldots, x_{i-1}, y_i, x_{i+1}, \ldots, x_n)\right|&\leq c.
\end{eqnarray*}
Then we have
$$\Pr\left(\left|f(\mathbf{X})-E[f(\mathbf{X})]\right|\geq \delta\right)\leq 2 e^{-\frac{2\delta^2}{nc^2}}.$$
\label{lem: AH}
\end{lemma}
\begin{proof}
A detailed proof can be found in \cite{MitUpf_05}.
\end{proof}

\section*{Appendix B: Properties of Random Walk}
Consider following two random walks.
\begin{enumerate}
\item[(1)] One dimensional random walk: A random walk on a circle with unit length and
$1/\tilde{S}$ points. At each time slot, a node moves to one point left, one point right or doesn't
move with equal probability as in Figure \ref{fig: 1RW}.
\begin{figure}[hbt]
\centering{\epsfig{file=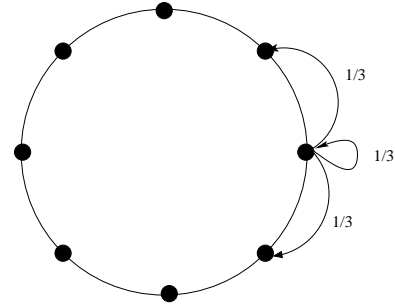,width=2.0in}} \caption{One Dimensional Random Walk} \label{fig:
1RW}
\end{figure}

\item[(2)] Two dimensional random walk: A random walk on a unit torus with $1/\tilde{S}^2$ points.
At each time slot, a node moves to one of eight neighbors or doesn't move with equal probability as
in Figure \ref{fig: 2RW}.
\end{enumerate}
\begin{figure}[hbt]
\centering{\epsfig{file=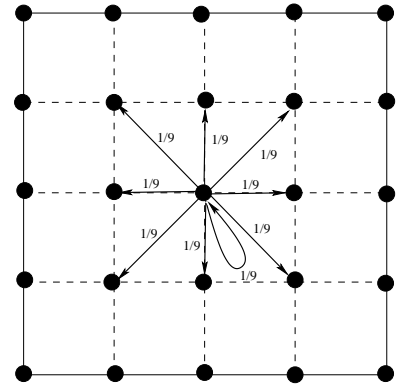,width=2.0in}} \caption{Two Dimensional Random Walk} \label{fig:
2RW}
\end{figure}
We introduce following definitions.
\begin{itemize}
\item Transition matrix $\mathbf{P}:$ $\mathbf{P}=\left[P_{\mathbf{i},\mathbf{j}}\right]$ where
$P_{\mathbf{i},\mathbf{j}}$ is the probability of moving from point $\mathbf{i}$ to point
$\mathbf{j}.$

\item Stationary distribution $\Pi:$ A vector which satisfies the equation $\Pi \mathbf{P}=\Pi.$

\item Hitting time $T_h(\mathbf{i}, \mathbf{j}):$ Time taken for a node to move from point
$\mathbf{i}$ to point $\mathbf{j}.$

\item Mixing time $T_m:$ $$T_m=\inf_t \sup_{\mathbf{i}}\sum_{\mathbf{j}}
\left|P^t_{\mathbf{i}\mathbf{j}}-\Pi_{\mathbf{j}}\right|\leq {\tilde{S}}^4,$$ where
$P^t_{\mathbf{i}\mathbf{j}}$ is the $(\mathbf{i}, \mathbf{j})^{\rm th}$ entry of $\mathbf{P}^t.$
\end{itemize}

\begin{lemma}
For the one dimensional random walk, we have
\begin{itemize}
\item $E[T_h(\mathbf{i}, \mathbf{j})]=O(1/\tilde{S}^2).$

\item $T_m=O(|\log \tilde{S}|/\tilde{S}^2).$
\end{itemize}

For the two dimensional random walk, we have
\begin{itemize}
\item $E[T_h(\mathbf{i}, \mathbf{j})]=O(|\log \tilde{S}|/\tilde{S}^2).$

\item $T_m=O(|\log \tilde{S}|/\tilde{S}^2).$
\end{itemize}
\label{lem: PRW}
\end{lemma}
\begin{proof}
Please refer to \cite{Lov_96} for the one dimensional random walk. The proof of the hitting time of
the two dimensional random walk is presented Lemma 13 in \cite{GamMamParSha_06}, and the mixing
time result holds since the two dimensional random walk can be regarded as two independent one
dimensional random walks.
\end{proof}

\begin{lemma}
Let $N_{\mathbf{i}, \mathbf{j}, \mathbf{k}}[D]$  denote the number of visits to point $\mathbf{j}$
in $D$ time slots starting from point $\mathbf{i}$ and ending at point $\mathbf{k}.$ If
$D=\omega(|\log \hat{S}|/\hat{S}),$ we have \begin{eqnarray}\frac{9}{10}D\hat{S}\leq
E[N_{\mathbf{i}, \mathbf{j}, \mathbf{k}}[D]]\leq \frac{11}{10}D\hat{S}\label{eq: RW_mean}
\end{eqnarray} where $\hat{S}=\tilde{S}$ for the one dimensional random walk and
$\hat{S}=\tilde{S}^2$ for the two dimensional random walk. Furthermore, if $D=\kappa\alpha|\log
\hat{S}|/\hat{S}^3$ where $\kappa=\omega(1),$ then we have
\begin{eqnarray}
\Pr\left(\frac{6}{5}D\hat{S}\geq N_{\mathbf{i}, \mathbf{j}, \mathbf{k}}[D]]\geq
\frac{4}{5}D\hat{S}\right)\geq 1- 2e^{-\frac{\alpha}{625}}. \label{eq: exact}
\end{eqnarray}
\label{lem: RWM}
\end{lemma}
\begin{proof}
First we have
$$T_h(\mathbf{i}, \mathbf{j})+\sum_{l=1}^{N_{\mathbf{i}, \mathbf{j}, \mathbf{k}}[D]-1}T^l_h(\mathbf{j},
\mathbf{j})+T_h(\mathbf{j}, \mathbf{k})=D,$$ where $T^l_h(\mathbf{j}, \mathbf{j})$ is the time
duration between $l^{\rm th}$ visits to point $\mathbf{j}$ and $(l+1)^{\rm th}$ visits to point
$\mathbf{j}.$ Taking the expectation on both sides, we have $$E[T_h(\mathbf{i},
\mathbf{j})]+E[N_{\mathbf{i}, \mathbf{j}, \mathbf{k}}[D]]E[T_h(\mathbf{j},
\mathbf{j})]-E[T_h(\mathbf{j}, \mathbf{j})]+E[T_h(\mathbf{j}, \mathbf{k})]=D,$$ which implies
$$E[N_{\mathbf{i}, \mathbf{j}, \mathbf{k}}[D]]= \frac{D-E[T_h(\mathbf{i},
\mathbf{j})]-E[T_h(\mathbf{j}, \mathbf{k})]+E[T_h(\mathbf{j}, \mathbf{j})]}{E[T_h(\mathbf{j},
\mathbf{j})]}.$$ Inequality (\ref{eq: RW_mean}) follows from the facts that $E[T_h(\mathbf{i},
\mathbf{j})]=O(|\log \hat{S}|/\hat{S})$ and $E[T_h(\mathbf{j}, \mathbf{j})]=1/\hat{S}.$

Next let $\mathbf{x}[t]$ denote the position of the node at time slot $t,$ $\tilde{\mathbf{X}}$
denote $\{\mathbf{x}[t]\}$ for $t=1, \frac{D{\hat{S}}^2}{\alpha}+1, \frac{2D{\hat{S}}^2}{\alpha}+1,
\ldots, D-\frac{D{\hat{S}}^2}{\alpha}, D,$ and $\mathbf{X}_m$ denote $\{\mathbf{x}[t]\}$ for
$t=\frac{mD{\hat{S}}^2}{\alpha}+1, \frac{mD{\hat{S}}^2}{\alpha}, \ldots, \min\left\{D,
\frac{(m+1)D{\hat{S}}^2}{\alpha}\right\}$ where $m=0, \ldots, \alpha/{\hat{S}}^2-1.$ Further, let
$N_{\mathbf{j} ,\tilde{\mathbf{X}}}[D]$ denote the number of visits to point $\mathbf{j}$ given
$\tilde{\mathbf{X}}.$ It is easy to see that for any $\tilde{\mathbf{X}}$ there exists a function
$f_{\tilde{\mathbf{X}}}$ such that
$$N_{\mathbf{j}, \tilde{\mathbf{X}}}[D]=f_{\tilde{\mathbf{X}}}\left(\mathbf{X}_1,\ldots,
\mathbf{X}_{\alpha/{\hat{S}}^2-1}\right),$$ where $\{\mathbf{X}_m\}$ are mutually independent given
$\tilde{\mathbf{X}}.$ Note that $\mathbf{X}_m$ contains the position information from time slot
$mD{\hat{S}}^2/\alpha+1$ to time slot $(m+1)D{\hat{S}}^2/\alpha,$ so
\begin{eqnarray*}\left|f_{\tilde{\mathbf{X}}}\left(\mathbf{X}_0, \ldots, \mathbf{X}_{m-1}, \mathbf{X}_m, \mathbf{X}_{m+1},\ldots, \mathbf{X}_{\alpha/{\hat{S}}^2-1}\right)
-\right.&\\\left. f_{\tilde{\mathbf{X}}}\left(\mathbf{X}_0, \ldots, \mathbf{X}_{m-1}, \mathbf{Y}_m,
\mathbf{X}_{m+1},\ldots, \mathbf{X}_{\alpha/{\hat{S}}^2 -1}\right)\right|&\leq
\frac{D{\hat{S}}^2}{\alpha}.\end{eqnarray*}

Next note that $D{\hat{S}}^2/{\alpha}=\omega(|\log{\hat{S}}|/\hat{S}),$ so from inequality
(\ref{eq: RW_mean}), we can conclude that for any $\mathbf{i},$ $\mathbf{j}$ and $\mathbf{k},$
$$\frac{9}{10}\frac{D{\hat{S}}^3}{\alpha}\leq E\left[N_{\mathbf{i}, \mathbf{j}, \mathbf{k}}\left[\frac{D{\hat{S}}^2}{\alpha}\right]\right]\leq
\frac{11}{10}\frac{D{\hat{S}}^3}{\alpha},$$ which implies that
$$\frac{9}{10}D{\hat{S}}\leq E\left[N_{\mathbf{j}, \tilde{\mathbf{\mathbf{X}}}}\left[D\right]\right]\leq
\frac{11}{10}D{\hat{S}}$$ holds for any $\tilde{\mathbf{X}}.$ Then from the Azuma-Hoeffding
inequality (Lemma \ref{lem: AH}), we have that
\begin{eqnarray*}
\Pr\left(\left|N_{\mathbf{j}, \tilde{\mathbf{X}}}[D]- E[N_{\mathbf{j},
\tilde{\mathbf{X}}}[D]]\right|\leq \frac{9}{100}D{\hat{S}}\right)\geq 1-2e^{-\frac{\alpha}{625}},
\end{eqnarray*} for any $\tilde{\mathbf{X}},$ and inequality (\ref{eq: exact}) holds.
\end{proof}

\end{document}